\documentclass[12pt]{amsart}
\usepackage{amsmath, amsthm, latexsym, amssymb, amsfonts, epsfig, epsf, xcolor, hyperref}
\usepackage{mathtools}
\usepackage{bm,xcolor,tikz,hyperref}
\usepackage{bm}
\usepackage{blkarray}
\usepackage{booktabs}

\theoremstyle{plain}
\newtheorem{theorem}{Theorem}[section]
\newtheorem{lemma}[theorem]{Lemma}
\newtheorem{corollary}[theorem]{Corollary}
\newtheorem{definition}[theorem]{Definition}
\newtheorem{proposition}[theorem]{Proposition}
\newtheorem{example}[theorem]{Example}
\newtheorem{remark}[theorem]{Remark}

\newcommand{\rmv}[1]{}

\newcommand{\F}{\mathbb F}

\begin{document}


\title{Polar decreasing monomial-Cartesian codes}

\author{Eduardo Camps}
\address[Eduardo Camps]{Escuela Superior de F\'isica y Matem\'aticas \\ Instituto Polit\'ecnico Nacional\\ Mexico City, Mexico}
\email{camps@esfm.ipn.mx}
\author{Hiram H. L\'opez}
\address[Hiram H. L\'opez]{Department of Mathematics and Statistics\\ Cleveland State University\\ Cleveland, OH USA}
\email{h.lopezvaldez@csuohio.edu}
\author{Gretchen L. Matthews}
\address[Gretchen L. Matthews]{Department of Mathematics\\ Virginia Tech\\ Blacksburg, VA USA}
\email{gmatthews@vt.edu}
\thanks{The first and fourth author are partially supported by SIP-IPN, project 20195717, and CONACyT. The third author was supported by NSF DMS-1855136.}
\author{Eliseo Sarmiento}
\address[Eliseo Sarmiento]{Escuela Superior de F\'isica y Matem\'aticas \\ Instituto Polit\'ecnico Nacional\\ Mexico City, Mexico}
\email{esarmiento@ipn.mx}
\thanks{}

\begin{abstract}
We  prove that families of polar codes with multiple kernels over certain symmetric channels can be viewed as polar decreasing monomial-Cartesian codes, offering a unified treatment for such codes, over any finite field. We define decreasing monomial-Cartesian codes as the evaluation of a set of monomials closed under divisibility over a Cartesian product. Polar decreasing monomial-Cartesian codes are  decreasing monomial-Cartesian codes whose sets of monomials are closed respect a partial order inspired by the recent work of Bardet, Dragoi, Otmani, and Tillich [``Algebraic properties of polar codes from a new polynomial formalism,'' 2016 IEEE International Symposium on Information Theory (ISIT)]. Extending the main theorem of Mori and Tanaka [``Source and Channel Polarization Over Finite Fields and Reed-Solomon Matrices," in IEEE Transactions on Information Theory, vol. 60, no. 5, pp. 2720--2736, May 2014], we prove that any sequence of invertible matrices over an arbitrary field satisfying certain conditions polarizes any symmetric over the field channel. In addition, we prove that the dual of a decreasing monomial-Cartesian code is monomially equivalent to a decreasing monomial-Cartesian code. Defining the minimal generating set for a set of monomials, we use it to describe the length, dimension and  minimum distance of a decreasing monomial-Cartesian code.

\rmv{Extending the main theorem of Mori and Tanaka [Source and channel polarization over finite fields and Reed-Solomon matrices. IEEE Transactions on Information Theory, {\bf {60}} (2014)], we prove that any sequence of invertible matrices over an arbitrary field satisfying certain conditions polarizes any symmetric over the field channel. We use these sequences of matrices to define a larger family of polar codes. The kernels arise from evaluation codes defined by a set of monomials which is closed under divisibility and evaluated over a Cartesian product, called decreasing monomial-Cartesian codes. We define a polar decreasing monomial-Cartesian code as given by a decreasing monomial-Cartesian code whose set of monomials is closed respect a 
partial order on a set of monomials, inspired by recent work of Bardet, Dragoi, Otmani, and Tillich [Algebraic properties of polar codes from a new polynomial formalism, In 2016 IEEE International Symposium on Information Theory (2016, July)]. }
\end{abstract}

\maketitle

\section{Introduction}\label{19.08.16}

Polar codes, introduced in 2009 in the seminal paper~\cite{Arikan} by Arikan, are the first class of provably capacity achieving codes for symmetric binary-input memoryless channels with explicit construction as well as efficient encoding and decoding. This breakthrough generated a flurry of activity on polar codes, as described below. Polar codes are now attracting increased attention as they are adopted in 5th generation wireless systems (5G) standardization process of the 3rd generation
partnership project (3GPP); for an overview, see for instance, \cite{5g1, 5g2}.

Originally they were constructed with Arikan's kernel, which is given by
\[G_A=\begin{bmatrix} 1&0\\ 1&1\end{bmatrix}.\]
The kernel is used to create $N$ synthetic channels from $N$ copies of the channel in a recursive fashion, so that some of the new channels have enhanced reliability while others are inferior. In the limit, as $N \rightarrow \infty$, each channel becomes either noiseless or pure noise, which is the so-called polarization phenomenon.  For an $(N,K)$ polar code, communication takes places over the $K$ most reliable channels, taking the corresponding codeword coordinates to be part of the information set while the remaining positions are frozen bits and not used to transfer information. 

Polar codes were generalized to arbitrary discrete memoryless channels by \c{S}a\c{s}o\u{g}lu, Telatar and Arikan \cite{STA}, and Korada, \c{S}a\c{s}o\u{g}lu, and Urbanke considered larger binary matrices as kernels and considered the speed of polarization by introducing a quantity called the exponent \cite{KSU}. Polarization over nonbinary alphabets was studied by  \c{S}a\c{s}o\u{g}lu  \cite{sas} as were polar codes over arbitrary finite fields by Mori and Tanaka \cite{scp} (see also \cite{erpol} and \cite{cparb}). Tal and Vardy pushed forward the applicability of polar codes with their introduction of a   successive-cancellation list decoder \cite{TV} (see also \cite{Sarkis}) and efficient constructions \cite{TV}. 

In this paper, we consider \textit{multikernel polar codes} where the kernel is formed using a sequence of matrices. The primary motivation for the  multikernel polarization process is the construction of polar codes of different lengths, other than $N=l^n$. Other  techniques, such as puncturing or shortening the original polar code (\cite{punc1}, \cite{punc2}, \cite{punc3}), have been employed to achieve this but with some disadvantages as augmenting the decoding complexity. Multikernel polar codes over the binary field were considered in \cite{mk2} and \cite{mk1}  where they give some conditions for a sequence of matrices to polarize a channel.  The paper is organized as follows. 

In Section \ref{19.12.03} we recover the definition of multikernel polarization given in \cite{mk1} with a slight difference, as well as define it for matrices and channels over non-binary fields. Taking the ideas of \cite{scp}, we focus on channel with a certain symmetry to describe when a sequence of square invertible matrices polarizes. This yields conditions which are easier to check than those given in \cite{mk2} for binary polar codes. Later in the paper, we delve into this setting to obtain polar decreasing monomial-Cartesian codes which arise from evaluation codes defined by monomials over finite fields (of any characteristic). 

Section  \ref{19.12.02} presents the decreasing monomial-Cartesian codes which are a particular class of evaluation codes. Evaluation codes form an important family of error-correcting codes, including Cartesian codes, algebraic geometry codes, and many variants finely tuned for specific applications, such as LCD codes, quantum codes, and locally recoverable codes \cite{lopez-matt-sopru}. In this paper, we consider evaluation codes formed by evaluating a set of monomials closed under divisibility at points in a Cartesian product. 
Decreasing monomial-Cartesian codes generalize Reed-Solmon and Reed-Muller codes, as we will see. In addition, they contain the family of decreasing monomial codes considered in \cite{Bardet}.  We will demonstrate that duals of such codes are of the same type, determine bases for them, and examine their classical parameters (meaning length, dimension, and minimum distance). This is in preparation for the application to polar codes in the next section. 

In Section \ref{polar_dec_section}, we consider polar codes whose kernels are decreasing monomial-Cartesian codes, calling these polar decreasing monomial-Cartesian codes. In \cite{Bardet}, the authors proved that polar codes constructed from $G_A$ are polar decreasing monomial-Cartesian codes over the binary field. We extend this result to prove in Theorem~\ref{19.12.04} that polar codes constructed from a sequence of Reed-Solomon matrices using Definition~\ref{19.12.07} are polar decreasing monomial-Cartesian codes, and that any symmetric over the field channel is polarized by this sequence of Reed-Solomon matrices, providing a unified framework for this family of polar codes. Naturally, this holds at the cost of reducing the family of channels over which we can work, given the required symmetric condition. Section \ref{conclusion} provides a conclusion to this work. 

We close this section with a bit of notation that will be useful in the remainder of this paper. We will use $K^*:=K \setminus \{ 0 \}$ to denote the multiplicative group of a field $K$. The set of $m \times n$ matrices over a field $K$ is denoted $K^{m \times n}$. Given $M \in K^{m \times n}$, $Row_iM$ denotes the $i^{th}$ row of $M$ and $Col_jM$ denotes its $j^{th}$ column. For more information about coding theory, we recommend~\cite{MacWilliams-Sloane,van-lint}. For algebraic concepts not described here, we suggest to the reader~\cite{monalg}. 

\section{Polar codes defined by sequences of invertible matrices} \label{19.12.03}
Let $\mathbb{F}_q$ be a finite field with $q$ elements. Consider a discrete memoryless channel (DMC) $W:\mathbb{F}_q\rightarrow\mathcal{Y}$ with transition probabilities $W(y|x)$, $y\in\mathcal{Y},\ x\in\mathbb{F}_q$. For a sequence of invertible matrices $\{T_i\}_{i=1}^\infty$ where $T_i \in \F_q^{n_i \times n_i}$, define $G'_m$ as
$$G'_m=T_1\otimes T_2\otimes\cdots\otimes T_m,$$ where $\otimes$ is for the Kronecker product and 
$$G_m=B_mG'_m$$
where $B_m$ is the permutation matrix that sends the row $\displaystyle j=k_m+\sum_{i=1}^{m-1}k_in_{i+1}$ to the row $\displaystyle  j'=k_1+\sum_{i=2}^m k_in_{i-1}$. Alternatively, we may define these matrices inductively, taking $G_1=T_1$ and for $m \geq 2$, $$G_m=\begin{bmatrix}
G_{m-1}\otimes Row_1T_m\\
G_{m-1}\otimes Row_2T_m\\
\vdots\\
G_{m-1}\otimes Row_{l_m}T_m\end{bmatrix}.$$ 

\begin{example}\rm
  Let $\alpha$ be a primitive element of $\mathbb{F}_4$. Over this field, take the next matrices:
$$T_1=\begin{bmatrix}
0&1&\alpha^2\\
0&1&\alpha\\
1&1&1\end{bmatrix},\hspace{.5cm}T_2=\begin{bmatrix}
0&1&1&1\\
0&1&\alpha^2&\alpha\\
0&1&\alpha&\alpha^2\\
1&1&1&1\end{bmatrix}$$
Then$$G'_2=\left[\begin{array}{cccccccccccc}0&     0&     0&     0&     0&     1&     1&     1&     0& \alpha^2&     \alpha^2&     \alpha^2\\
     0&     0&     0&     0&     0&     1&    \alpha^2&     \alpha&     0&     \alpha^2&     \alpha&     1\\
     0&     0&     0&     0&     0&     1&     \alpha&     \alpha^2&     0&     \alpha^2&     1&     \alpha\\
     0&     0&     0&     0&     1&     1&     1&     1&     \alpha^2&     \alpha^2&     \alpha^2&
     \alpha^2\\
     0&     0&     0&     0&     0&     1&    1&     1&     0&     \alpha&     \alpha&     \alpha\\
     0&     0&     0&     0&     0&     1&     \alpha^2&     \alpha&     0&     \alpha&     1&    \alpha^2\\
     0&     0&     0&     0&     0&     1&     \alpha&     \alpha^2&     0&     \alpha&     \alpha^2&     1\\
     0&     0&     0&     0&     1&     1&     1&     1&     \alpha&     \alpha&     \alpha&     \alpha\\
     0&     1&     1&     1&     0&     1&     1&     1&     0&     1&     1&     1\\
     0&     1&     \alpha^2&     \alpha&     0&     1&     \alpha^2&     \alpha&     0&     1&     \alpha^2&     \alpha\\
     0&     1&     \alpha&     \alpha^2&     0&     1&     \alpha&     \alpha^2&     0&     1&     \alpha&     \alpha^2\\
     1&     1&     1&     1&     1&     1&     1&     1&     1&     1&1&1\\
     \end{array}\right],$$
     \noindent and
     $$B_2=\left[\begin{array}{ccccccccccccc}
     1&      0&      0&      0&      0&      0&      0&      0&      0&      0&      0&      0\\      0&0 & 0&      0&      1&      0&      0&      0&      0&      0&      0&      0\\      0&      0&      0&      0&      0&      0&      0&      0&      1&      0&      0&      0\\      0&      1&      0&      0&      0&      0&      0&      0&      0&      0&      0&      0\\      0&      0&      0&      0&      0&      1&      0&      0&      0&      0&      0&      0\\      0&      0&      0&      0&      0&      0&      0&      0&      0&      1&      0&      0\\      0&      0&      1&      0&      0&      0&      0&      0&      0&      0&      0&      0\\      0&      0&      0&      0&      0&      0&      1&      0&      0&      0&      0&      0\\      0&      0&      0&      0&      0&      0&      0&      0&      0&      0&      1&      0\\      0&      0&      0&      1&      0&      0&      0&      0&      0&      0&      0&      0\\      0&      0&      0&      0&      0&      0&      0&      1&      0&      0&      0&      0\\      0&      0&      0&      0&      0&      0&      0&      0&      0&      0&0&1\end{array}\right].$$
Therefore,
$$G_2=\left[\begin{array}{cccccccccccc}
0&      0&      0&      0&      0&      1&      1&      1&      0&      \alpha^2&      \alpha^2&      \alpha^2\\
0&      0&      0&      0&      0&      1&      1&      1&      0&      \alpha&      \alpha&      \alpha\\
0&      1&      1&      1&      0&      1&      1&      1&      0&      1&      1&      1\\
0&      0&      0&      0&      0&      1&      \alpha^2&      \alpha&      0&      \alpha^2&      \alpha&      1\\
0&      0&      0&      0&      0&      1&      \alpha^2&      \alpha&      0&      \alpha&      1&      \alpha^2\\
0&      1&      \alpha^2&      \alpha&      0&      1&      \alpha^2&      \alpha&      0&      1&      \alpha^2&      \alpha\\
0&      0&      0&      0&      0&      1&      \alpha&      \alpha^2&      0&      \alpha^2&      1&      \alpha\\
0&      0&      0&      0&      0&      1&      \alpha&      \alpha^2&      0&      \alpha&      \alpha^2&      1\\
0&      1&      \alpha&      \alpha^2&      0&      1&      \alpha&      \alpha^2&      0&      1&      \alpha&      \alpha^2\\
0&      0&      0&      0&      1&      1&      1&      1&      \alpha^2&      \alpha^2&      \alpha^2&      \alpha^2\\
0&      0&      0&      0&      1&      1&      1&      1&      \alpha&      \alpha&      \alpha&      \alpha\\
1&      1&      1&      1&      1&      1&      1&      1&      1&      1&      1&      1\end{array}\right].$$
\end{example}

Let us continue with the description of polarization. Starting from the channel $W$, we construct the following $n=\prod_{i=1}^m n_i$ channels:
$$W_m^{(i)}:\mathbb{F}_q\rightarrow\mathcal{Y}^N\times\mathbb{F}_q^{i-1}$$
$$W_m^{(i)}\left(y_1^n,u_1^{i-1}|u_i\right)=\frac{1}{q^{n-1}}\sum_{u_{i+1}^n\in\mathbb{F}_q^{n-1}}\prod_{j=1}^n W\left(y_j|u_1^n Col_j(G_m)_{\ast}\right).$$

As $n$ grows, some of the channels $W_m^{(i)}$ becomes noiseless. We measure this through the symmetric rate of the channel.

\begin{definition}\rm
    Let $W:\F_q \rightarrow\mathcal{Y}$ be a DMC channel. We define the symmetric rate of $W$ as
    $$I(W)=\frac{1}{q}\sum_{(x,y)\in\F_q\times\mathcal{Y}}W(y|x)\log_q\left(\frac{W(y|x)}{\frac{1}{q}\sum_{x\in\mathcal{X}}W(y|x)}\right).$$
\end{definition}

 \begin{definition}\rm
    Let $W:\mathbb{F}_q\rightarrow\mathcal{Y}$ be a DMC channel and $\{T_i\}_{i=1}^\infty$ be a sequence of invertible matrices over $\mathbb{F}_q$. We say that the sequence polarizes $W$ if for each $\delta>0$, we have 
    $$\lim_{m\rightarrow\infty}\frac{\left|\left\{i\in\{1,\ldots,\prod_{i=1}^m n_i\}\ |\ I\left(W_m^{(i)}\right)\in(1-\delta,1]\right\}\right|}{\prod_{i=1}^m n_i}=I(W),$$ and 
    $$\lim_{m\rightarrow\infty}\frac{\left|\left\{i\in\{1,\ldots,\prod_{i=1}^m n_i\}\ |\ I\left(W_m^{(i)}\right)\in[0,\delta)\right\}\right|}{\prod_{i=1}^m n_i}=1-I(W).$$
 \end{definition}
 Observe that when $T_i=G$ for all $i$, then we have the usual polarization process with kernel $G$. By taking $T_i=G_A$ for all $i,$ we have the original polar code defined by Arikan. The previous definition is similar to that given in \cite{mk1}, with the difference being we use the bit-reversal matrix $B_m$ and the field $\mathbb{F}_q$ instead of $\mathbb{F}_2.$

\begin{definition}\rm
    Let $W:\mathbb{F}_q\rightarrow\mathcal{Y}$ be a DMC channel. Then:
    \begin{itemize}
        \item[(a)]  $W$ is symmetric over the sum or additive symmetric if for each $a\in\mathbb{F}_q$ there is a permutation $\sigma_a$ of $\mathcal{Y}$ such that
        $$W(y|x)=W(\sigma_a(y)|x+a),\ \ \ \ \forall x\in\mathbb{F}_q,\ y\in\mathcal{Y}.$$
        
        \item[(b)]  $W$ is symmetric over the product if for each $a\in\mathbb{F}_q^\ast$ there is a permutation $\psi_a$ of $\mathcal{Y}$ such that
        $$W(y|x)=W(\psi_a(y)|ax),\ \ \ \ \forall x\in\mathbb{F}_q,\ y\in\mathcal{Y}.$$
        
        \item[(c)]  $W$ is symmetric over the field (SOF) if it is both symmetric over the sum and over the product.
    \end{itemize}
\end{definition}

Originally, polar codes were proposed over binary symmetric channels \cite{Arikan}. Later, in \cite{scp}, symmetry over the sum was used to guarantee that a family of matrices polarizes such channels. In \cite{Vardohus}, the authors employed symmetry over the field to describe up to certain degree the best channels $W_n^{(i)}$; these are those with greater symmetric rate. A channel with symmetry over both sums and products is called symmetric over the field (SOF), as detailed above.

\begin{example}\rm
    Let $0\leq p\leq 1$. The $q$-ary symmetric channel is defined as
    $$W_{Sq}:\mathbb{F}_q\rightarrow\mathbb{F}_q$$
    $$W_{Sq}(y|x)=(1-p)\delta(x,y)+\frac{p}{q},$$ where $\delta(x,y)=1$ if $x=y$ and $0$ otherwise. This is a SOF channel.
\end{example}

\begin{example}\rm
    The $q$-ary erasure channel for $0\leq p\leq 1$ is defined as $$W_{qE}:\mathbb{F}_q\rightarrow\mathbb{F}_q\cup\{\ast\}$$ with transition probabilities
    $$W_{qE}(y|x)=\begin{cases}
    1-p& y=\ast\\
    p& y=x\\
    0&\text{otherwise}\end{cases}.$$
    
    This is a SOF channel. The polar behavior of generalized Reed-Solomon codes over this channel was studied in \cite{erpol}.
\end{example}

When $W$ is an additive symmetric channel and $G$ and $G'$ are invertible matrices such that $G'G^{-1}$ is an upper-triangular matrix, then using either $G$ or $G'$ to polarize gives rise on channels $W_1^{(i)}$ with same symmetric rate. If $G$ polarizes, then $G'$ polarizes $W$. Making a column permutation of $G$ does not affect  the symmetric rate of the channels. If $P$ is a permutation matrix and $G$ polarizes, then so does $GP.$ This leads to the following definition.
\begin{definition}\rm
Let $G \in \F_q^{l \times l}$ be invertible. Let $V \in \F_q^{l \times l}$ be an upper-triangular invertible matrix and $P \in \F_q^{l \times l}$ be a permutation matrix. If $G'=VGP$ is a lower-triangular matrix with $1$'s in its diagonal, then $G'$ is called a {\bf standard form} of $G$.
\end{definition}

It is important to note that standard form is not unique. Over $\mathbb{F}_4$ with primitive element $\alpha$, both \[G'_1=\begin{bmatrix} 
\alpha&\alpha^2\\
0&\alpha\end{bmatrix}G=\begin{bmatrix}1&0\\ \alpha&1\end{bmatrix} \text{ and } 
G'_2=\begin{bmatrix}
\alpha^2&\alpha^2\\
0&1\end{bmatrix}G\begin{bmatrix} 0&1\\ 1&0\end{bmatrix}=\begin{bmatrix} 1&0\\ \alpha^2&1\end{bmatrix}\]
are standard forms of $G=\begin{bmatrix} 1&1\\ 1&\alpha^2\end{bmatrix}.$ The information given by the standard form of a sequence of invertible matrices is enough to determine if such a sequence polarizes an additive symmetric channel.

\begin{lemma}\cite[Theorem 14]{scp}\label{19.12.06}
    Let $p$ a prime such that $p|q$. The followings are equivalent for an invertible matrix $G \in \F_q^{l \times l}$ with a non identity standard form.
    
    \begin{itemize}
        \item[\rm {(a)}] Any additive symmetric channel is polarized by $G$.
        
        \item[\rm {(b)}] The field extension of $\F_p$ generated by the entries of $G'$, denoted $\mathbb{F}_p(G')$, is $\F_q$
         for any standard form $G'$ of $G$; that is,                 $$\mathbb{F}_p(G')=\mathbb{F}_q$$  for any standard form $G'$ of $G$.
        \item[\rm {(c)}] There is a standard form $G'$ of $G$ with $\mathbb{F}_p(G')=\mathbb{F}_q$. 
    \end{itemize}
\end{lemma}

\begin{theorem}\label{19.12.01}
    Let $\{T_i\}_{i=1}^\infty$ be a sequence of invertible matrices. If for each $i$, $T_i$ has a non identity standard form $T'_i$ such that $\mathbb{F}_p(T'_i)=\mathbb{F}_q$, then the sequence $\{T_i\}_{i=1}^\infty$ polarizes to any additive symmetric channel $W$.
\end{theorem}
\begin{proof}
The proof of the sufficency of Lemma~\ref{19.12.06} relies on the fact that the process $I\left(W_m^{(i)}\right)$ forms a martingale and this channels are as good as $$\left(W_m^{(i)}\right)_1^{(2)},$$ where the last is the second splitted channel by using any $G_\gamma=\begin{bmatrix} 1&0\\ \gamma&1\end{bmatrix}$. The same arguments apply here with slight changes to the process by substituting the sequence $\{G\}_{i=1}^\infty$ by any other sequence $\{T_i\}_{i=1}^\infty$ of invertible matrices.
\end{proof}
The previous result does not imply that if a sequence $\{T_i\}_{i=1}^\infty$ polarizes, then each $T_i$ has a non identity standard form $T'_i$ with $\mathbb{F}_p(T'_i)=\mathbb{F}_q$.  It is enough to consider a sequence $\{I_l\}\cup\{T_i\}_{i=1}^\infty$, where $I_l$ is the identity matrix of size $l$ and each $T_i$ has a non identity standard form with the condition asked before.

In \cite{mk2}, the authors gave conditions over $\mathbb{F}_2$ for a sequence to polarize. Since we are interested on SOF channels, we can strength the last proposition to the following result.

\begin{corollary}
Let $\{T_i\}_{i=1}^\infty$ be a sequence of invertible matrices. If for each $i$, $T_i$ has a non identity standard form, then the sequence $\{T_i\}_{i=1}^\infty$ polarizes any SOF channel $W$.
\end{corollary}

The proof of the last relies in the following lemma.

\begin{lemma}
    Let $G\in \F_q^{l\times l}$ be an invertible matrix and $G'$ be the matrix with $Col_1G'=aCol_1G$ for some $a \in \F_q^*$ and $Col_jG'=Col_jG$ for $2 \leq j \leq n$. Let $W:\mathbb{F}_q\rightarrow\mathcal{Y}$ be a SOF channel. If $W_1^{(i)}$, $1\leq i\leq l$ are the splitted channels of the polarization process using $G$ and ${W'}_1^{(i)}$, $1\leq i\leq l$ are the same but with $G'$, then
    $$I\left(W_1^{(i)}\right)=I\left({W'}_1^{(i)}\right).$$
\end{lemma}

\begin{proof}
    Let $\psi_a$ the permutation of $\mathcal{Y}$ such that $$W(y|x)=W(\psi_a(y)|ax)$$ for any $x\in\mathbb{F}_q$ and $y\in\mathcal{Y}$. Then

    \begin{align*}
    W_1^{(i)}\left(y_1^l,u_1^{i-1}|u_i\right)=&\sum_{u_{i+1}^l\in\mathbb{F}_q^{l-1}}\prod_{j=1}^l W\left(y_j|u_1^l Col_j G\right)\\
    =&\sum_{u_{i+1}^l\in\mathbb{F}_q^{l-1}}\left(W(\psi_a(y_1)|u_1^l(aCol_1 G))\prod_{j=2}^l W\left(y_j|u_1^l Col_j G\right)\right)\\
    =&{W'}_1^{(i)}\left((\psi_a(y_1),y_2^l),u_1^{i-1}|u_i\right)\end{align*}
    
    Since $W_1^{(i)}$ and ${W'}_1^{(i)}$ has the same distribution but a bijection over the output alphabet, they have the same symmetric rate. 
\end{proof}

If $T_i$ has a non identity standard form, we can multiply the $n_i-1$ column by some $a\in\mathbb{F}_q^\ast$ to obtain $\overline{T}_i$ which has a standard $\overline{T}'_i$ form such that $\mathbb{F}_p(\overline{T}'_i)=\mathbb{F}_q$. Since a SOF channel is symmetric, the sequence $\{\overline{T}_i\}_{i=1}^\infty$ polarizes and by the last lemma, $\{T_i\}_{i=1}^\infty$ polarizes too. In the light of this, we can generalize the definition of polar codes.

\begin{definition}\rm
Let $W:\mathcal{X}\rightarrow\mathcal{Y}$ be a DMC channel with $|\mathcal{X}|=q$. For $x,x'\in\mathbb{F}_q$, $x\neq x'$, we define the {\bf Bhattacharyya distance} as
\[Z(x,x')=\sum_{y\in\mathcal{Y}}\sqrt{W(y|x)W(y|x')}\]
and the {\bf Bhattacharyya parameter} as
\[Z(W)=\frac{1}{q(q-1)}\sum_{\substack{x,x'\in\mathcal{X}\\ x\neq x'}}Z(x,x'),\]
 the average of the Bhattacharyya distances over $\mathcal{X}$.
\end{definition}

\begin{definition}\rm\label{19.12.07}
Let $\{T_i\}_{i=1}^\infty$ be a sequence of invertible matrices that polarizes the channel $W:\mathbb{F}_q\rightarrow\mathcal{Y}$. Let $m$ be a positive integer and let $n=\prod_{i=1}^m n_i$, where $n_i$ are the sizes of $T_i$ as before. We define an {\bf information set} $\mathcal{A}_m\subset\{1,\ldots,n\}$ as a set such that
\[Z\left(W_m^{(i)}\right)\leq Z\left(W_m^{(j)}\right),\ \ \ \forall i\in\mathcal{A}_m, \quad \forall j\notin\mathcal{A}_m.\]
A {\bf polar code} is the subspace $C_{\mathcal{A}_m}$ generated by the rows of $G_m$ indexed by $\mathcal{A}_m.$
\end{definition}
It is known that  $I(W)\rightarrow 1$ if and only if $Z(W)\rightarrow 0$ \cite[Lemma 5]{cparb}. Therefore, as $n$ grows, it is the same selecting $Z$ or $I$ to construct $\mathcal{A}_m$, but by selecting $Z$ we can easily (upper) bound the error probability for a successive cancellation decoder.

\section{Decreasing monomial-Cartesian codes}\label{19.12.02}

A decreasing monomial-Cartesian code is defined using the following concepts. Let $K:=\mathbb{F}_q$ be a finite field with $q$ elements and $R:=K[x_1,\ldots,x_m]$ be the polynomial ring over $K$ in $m$ variables. Given a point $\bm{a}=(a_1,\dots,a_m)\in\mathbb{Z}_{\geq 0}^m$,  $\bm{x}^{\bm{a}}$ is the corresponding monomial in $R$; i.e. $\bm{x}^{\bm{a}}:=x_1^{a_1}\cdots x_m^{a_m}.$ A {\bf decreasing monomial set} is a set of monomials $\mathcal{M}\subseteq R$ such that both conditions $M\in \mathcal{M}$ and $M^\prime$ divides $M$ imply $M^\prime \in \mathcal{M}.$ Let $L(\mathcal{M})$ be the subspace of polynomials of $R$ that are $K$-linear combinations of monomials of $\mathcal{M}:$
\[ L(\mathcal{M}):=\operatorname{Span}_K\{M : M \in \mathcal{M} \}\subseteq R.\]
Fix non-empty subsets $S_1,\ldots,S_m$ of $K$. The {Cartesian product} is defined by
\[\mathcal{S}:=S_1\times\cdots\times S_m\subseteq K^{m}.\]
In what follows, $n_i:=|S_i|$, the cardinality of $S_i$ for $i\in[m]:=\left\{ 1, \dots, m \right\}$, and $n:=|\mathcal{S}|,$ the cardinality of $\mathcal{S}.$
Fix a linear order on $\mathcal{S}=\{\bm{s}_1,\ldots,\bm{s}_n\},$ $\bm{s}_1 \prec \cdots \prec \bm{s}_n$. We define an {\bf evaluation map}
\[
\begin{array}{lclll}
{\operatorname{ev}_{\mathcal{S}}}\colon & {L}(\mathcal{M})\ & \to & K^n\\
&f &\mapsto &  \left( f(\bm{s}_1),\ldots,f(\bm{s}_n)\right).
\end{array}
\]
From now on, we assume that the degree of each monomial $M\in \mathcal{M}$ in $x_i$ is less than $n_i$. In this case the evaluation map $\operatorname{ev}_{\mathcal{S}}$ is injective, see \cite[Proposition 2.1]{lopez-matt-sopru}. The {\bf complement} of $\mathcal{M}$ in $\mathcal{S}$ denoted by $\mathcal{M}^c_{\mathcal{S}},$ is the set of all monomials in $R$ that are not in $\mathcal{M}$ and their degree respect $x_i$ is less than $n_i.$

\begin{definition}\rm Let $\mathcal{M} \subseteq R$ be a decreasing monomial set. The image $\operatorname{ev}_{\mathcal{S}}(L(\mathcal{M}))\subseteq K^n$ is called the {\bf decreasing monomial-Cartesian code} associated to $\mathcal{S}$ and $\mathcal{M}$. We denote it by $C(\mathcal{S},\mathcal{M})$. When the monomial set is not decreasing, the associated code is called {\bf monomial-Cartesian code}~\cite{lopez-matt-sopru}.
\end{definition}

The length and the dimension of a decreasing monomial-Cartesian code $C(\mathcal{S},\mathcal{M})$ are given by $n=|\mathcal{S}|$ and $k=\dim_K C(\mathcal{S},\mathcal{M})=|\mathcal{M}|$, respectively \cite[Proposition 2.1]{lopez-matt-sopru}. Recall that the {\bf minimum distance\/} of a code $C$  is given by 
\[
d(C)=\min\{|\operatorname{Supp}({\bm c})| : {\bm 0}\neq {\bm c}\in C\},
\]
where $\operatorname{Supp}({\bm c})$ denotes the support of ${\bm c}$, that is the set of all non-zero entries of ${\bm c}$. Unlike the case of the length and the dimension, in general, there is no explicit formula for $d (C(\mathcal{S},\mathcal{M}))$ in terms of $\mathcal{S}$ and $\mathcal{M}.$

The {\bf dual} of a code $C$ is defined by
\[C^{\perp} = \{ \bm{w}\in K^n : \bm{w}\cdot\bm{c}=0 \text{ for all } \bm{c}\in C \}, \]
where $\bm{w}\cdot\bm{c}$ represents the {\bf Euclidean inner product}. The code $C$ is called a {\bf linear complementary dual} ({\bf LCD}) \cite{JMassey} if $C\cap C^{\perp}=\{ \bm{0} \},$ and is called a {\bf self-orthogonal} code if $C^{\perp} \subseteq C.$

Instances of decreasing monomial-Cartesian codes for particular families of Cartesian products $\mathcal{S}$ and particular families of decreasing monomial sets $\mathcal{M}$ have been previously studied in the literature. For example, a {\bf Reed-Muller code} of order $r$ in the sense of  \cite[p.~37]{tsfasman} is a decreasing monomial-Cartesian code $C(K^m,M_r),$ where $M_r$ is the set of monomials of degree less than $r$.  An {\bf affine Cartesian code} of order $r$ is the {decreasing monomial-Cartesian code} $C(\mathcal{S},M_{r}).$ This family of affine Cartesian codes appeared first in \cite{Geil} and then independently in \cite{lopez-villa}. In \cite{Bardet}, the authors studied the case when the finite field $K$ is $\mathbb{F}_2$ and the set of monomials satisfy some decreasing conditions; then their results were generalized in \cite{Vardohus} for $K=\mathbb{F}_q$ and monomials associated to  curve kernels. The case when the set of monomials $\mathcal{M}$ is a tensor product, the minimum distance of the associated code can be computed using the same ideas that~\cite{Sop}.

A {\bf monomial matrix} is a square matrix with exactly one nonzero entry in each row and column.
Let $C_1$ and $C_2$ be codes of the same length over $K$, and let $G_1$ be a generator matrix for
$C_1.$ Then $C_1$ and $C_2$ are {\bf monomially equivalent} provided there is a monomial matrix $M$
with entries over the same field $K$ so that $G_1M$ is a generator matrix of $C_2.$ Monomially
equivalent codes have the same length, dimension, and minimum distance.

\begin{definition}\rm
For $\bm{s}=\left(s_1,\ldots,s_m\right)\in \mathcal{S}$ and
$f\in R,$ define the {\bf residue} of $f$ at $\bm{s}$ as
\[
\operatorname{Res}_{\bm{s}}f=f(\bm{s})
\left(
\prod_{i=1}^m \prod_{\substack{s_{i}^\prime\in S_i\setminus\{s_i\}}}\left(s_i-s_i^{\prime}\right)
\right)^{-1}.
\]
and the {\bf residue vector} of $f$ at $\mathcal{S}$ as
\[\operatorname{Res}_{\mathcal S}f=
\left(\operatorname{Res}_{\bm{s}_1}f,\ldots, \operatorname{Res}_{\bm{s}_n}f \right).\]
\end{definition}

\begin{theorem}
The dual of the code $C(\mathcal{S},\mathcal{M})$ is monomially equivalent to a
decreasing monomial-Cartesian code. In fact,
$$C(\mathcal{S},\mathcal{M})^{\perp} = \operatorname{Span}_K\left(\left\{\operatorname{Res}_{\mathcal{S}} \frac{x_1^{n_1-1}\cdots x_m^{n_m-1}}{M}: M\in \mathcal{M}^c\right\} \right).$$ Moreover, $$\Delta:=\left\{\operatorname{Res}_{\mathcal{S}} \frac{x_1^{n_1-1}\cdots x_m^{n_m-1}}{M}: M\in \mathcal{M}^c\right\}$$ is a basis for $C(\mathcal{S},\mathcal{M})^{\perp}$. 
\end{theorem}
\begin{proof}
We starting proving that the set
\[ \Delta^\prime:=\left\{\frac{x_1^{n_1-1}\cdots x_m^{n_m-1}}{M}: M\in \mathcal{M}^c_{\mathcal{S}}\right\} \]
is decreasing. Let $M \in \mathcal{M}^c_{\mathcal{S}}$ and $\bm{x}^{\bm a}$ a divisor of
$\frac{x_1^{n_1-1}\cdots x_m^{n_m-1}}{M}.$ Then there exists a monomial $\bm{x}^{\bm{b}}$ in $R$
such that $\frac{x_1^{n_1-1}\cdots x_m^{n_m-1}}{M}={\bm{x}^{\bm{a}}\bm{x}^{\bm{b}}}.$
As $M\in \mathcal{M}^c$
and $\mathcal{M}$ is decreasing, then ${\bm{x}^{\bm{b}}}M\in \mathcal{M}^c$ and
\(\bm{x}^{\bm{a}}=\frac{x_1^{n_1-1}\cdots x_m^{n_m-1}}{\bm{x}^{\bm{b}}M} \in \Delta^\prime.\)
This proves that the set $\Delta^\prime$ is decreasing.
Due to \cite[Theorem 2.7]{lopez-matt-sopru} and its proof, \(\Delta\) is a basis
for the dual $C(\mathcal{S},\mathcal{M})^{\perp}.$ Finally, it is clear that
$\operatorname{Span}_K\{{\bm c} : {\bm c} \in \Delta \}$
is monomially equivalent to $\operatorname{ev}_{\mathcal{S}}(\Delta^\prime),$ which is a
{decreasing monomial-Cartesian code}.
\end{proof}

\begin{example}\rm \label{19.08.17}
Let $K=\mathbb{F}_7,$ $\mathcal{S}=K^2$ and $\mathcal{M}$ the set of monomials of $K[x_1,x_2]$
whose exponents are the points in the left picture below. Then the code
$C(\mathcal{S},\mathcal{M})$ is generated by the vectors
$\operatorname{ev}_{\mathcal{S}} (\text{\textcolor{red}{$M$}}),$ where \textcolor{red}{$M$} is a monomial
whose exponent is a point in the left picture below and the dual
$C(\mathcal{S},\mathcal{M})^{\perp}$ is generated by the vectors
$\operatorname{Res}_{\mathcal{S}} (\text{\textcolor{blue}{$M$}}),$ where \textcolor{blue}{$M$} is a
monomial whose exponent is a point in the right picture in Figure \ref{fig1}.
\end{example}

\begin{figure}[h] \label{fig1}
\noindent
\begin{minipage}[t]{0.4\textwidth}
\begin{tikzpicture}[scale=0.7]
\draw [-latex] (0,0) -- (6.5,0)node[right] {$K$};
\draw [dashed] (0,1)node[left]{1} -- (6,1)node[right] {};
\draw [dashed] (0,2)node[left]{2} -- (6,2)node[right] {};
\draw [dashed] (0,3)node[left]{3} -- (6,3)node[right] {};
\draw [dashed] (0,4)node[left]{4} -- (6,4)node[right] {};
\draw [dashed] (0,5)node[left]{5} -- (6,5)node[right] {};
\draw [dashed] (0,6)node[left]{6} -- (6,6)node[right] {};
\draw [-latex] (0,0)node[below left]{0} -- (0,6.5)node[above] {$K$};
\draw [dashed] (1,0)node[below]{1} -- (1,6)node[right] {};
\draw [dashed] (2,0)node[below]{2} -- (2,6)node[right] {};
\draw [dashed] (3,0)node[below]{3} -- (3,6)node[right] {};
\draw [dashed] (4,0)node[below]{4} -- (4,6)node[right] {};
\draw [dashed] (5,0)node[below]{5} -- (5,6)node[right] {};
\draw [dashed] (6,0)node[below]{6} -- (6,6)node[right] {};
\fill [color=red](0,0) {circle(.10cm)};
\fill [color=red](1,0) {circle(.10cm)};
\fill [color=red](2,0) {circle(.10cm)};
\fill [color=red](3,0) {circle(.10cm)};
\fill [color=red](4,0) {circle(.10cm)};
\fill [color=red](5,0) {circle(.10cm)};
\fill [color=red](0,1) {circle(.10cm)};
\fill [color=red](1,1) {circle(.10cm)};
\fill [color=red](2,1) {circle(.10cm)};
\fill [color=red](3,1) {circle(.10cm)};
\fill [color=red](4,1) {circle(.10cm)};
\fill [color=red](5,1) {circle(.10cm)};
\fill [color=red](0,2) {circle(.10cm)};
\fill [color=red](1,2) {circle(.10cm)};
\fill [color=red](2,2) {circle(.10cm)};
\fill [color=red](3,2) {circle(.10cm)};
\fill [color=red](4,2) {circle(.10cm)};
\fill [color=red](5,2) {circle(.10cm)};
\fill [color=red](0,3) {circle(.10cm)};
\fill [color=red](1,3) {circle(.10cm)};
\fill [color=red](2,3) {circle(.10cm)};
\fill [color=red](3,3) {circle(.10cm)};
\fill [color=red](4,3) {circle(.10cm)};
\fill [color=red](0,4) {circle(.10cm)};
\fill [color=red](1,4) {circle(.10cm)};
\fill [color=red](2,4) {circle(.10cm)};
\fill [color=red](3,4) {circle(.10cm)};
\fill [color=red](4,4) {circle(.10cm)};
\fill [color=red](0,5) {circle(.10cm)};
\fill [color=red](1,5) {circle(.10cm)};
\fill [color=red](2,5) {circle(.10cm)};
\fill [color=red](0,6) {circle(.10cm)};
\fill [color=red](1,6) {circle(.10cm)};
\fill [color=red](2,6) {circle(.10cm)};
\end{tikzpicture}
\end{minipage}
\begin{minipage}[t]{0.4\textwidth}
\begin{tikzpicture}[scale=0.7]
\draw [-latex] (6,6) -- (6,-0.5)node[below] {$K$};
\draw [dashed] (0,0)node[left]{} -- (6,0)node[right] {6};
\draw [dashed] (0,1)node[left]{} -- (6,1)node[right] {5};
\draw [dashed] (0,2)node[left]{} -- (6,2)node[right] {4};
\draw [dashed] (0,3)node[left]{} -- (6,3)node[right] {3};
\draw [dashed] (0,4)node[left]{} -- (6,4)node[right] {2};
\draw [dashed] (0,5)node[left]{} -- (6,5)node[right] {1};
\draw [dashed] (0,6)node[left]{} -- (6,6)node[right] {};
\draw [-latex] (6,6)node[above right]{0} -- (-0.5,6)node[left] {$K$};
\draw [dashed] (0,0)node[below]{} -- (0,6)node[above] {6};
\draw [dashed] (1,0)node[below]{} -- (1,6)node[above] {5};
\draw [dashed] (2,0)node[below]{} -- (2,6)node[above] {4};
\draw [dashed] (3,0)node[below]{} -- (3,6)node[above] {3};
\draw [dashed] (4,0)node[below]{} -- (4,6)node[above] {2};
\draw [dashed] (5,0)node[below]{} -- (5,6)node[above] {1};
\draw [dashed] (6,0)node[below]{} -- (6,6)node[above] {};
\fill [color=blue](6,0) {circle(.10cm)};
\fill [color=blue](6,1) {circle(.10cm)};
\fill [color=blue](6,2) {circle(.10cm)};
\fill [color=blue](5,3) {circle(.10cm)};
\fill [color=blue](6,3) {circle(.10cm)};
\fill [color=blue](5,4) {circle(.10cm)};
\fill [color=blue](6,4) {circle(.10cm)};
\fill [color=blue](3,5) {circle(.10cm)};
\fill [color=blue](4,5) {circle(.10cm)};
\fill [color=blue](5,5) {circle(.10cm)};
\fill [color=blue](6,5) {circle(.10cm)};
\fill [color=blue](3,6) {circle(.10cm)};
\fill [color=blue](4,6) {circle(.10cm)};
\fill [color=blue](5,6) {circle(.10cm)};
\fill [color=blue](6,6) {circle(.10cm)};
\end{tikzpicture}
\end{minipage}
\caption{} 
\end{figure}

\begin{definition}\label{gset}\rm
A subset $\mathcal{B}({\mathcal{M}}) \subseteq \mathcal{M}$ is a {\bf generating set} of $\mathcal{M}$
if for every $M\in \mathcal{M}$ there exists a monomial $B\in \mathcal{B}({\mathcal{M}})$ such that
$M$ divides $B.$ A generating set $\mathcal{B}({\mathcal{M}})$ is called {\bf minimal} if for every
two elements $B_1, B_2 \in \mathcal{B}({\mathcal{M}}),$ $B_1$ does not divide $B_2$ and
$B_2$ does not divide $B_1.$
\end{definition}

\begin{example}\rm\label{19.08.20}
Let $K=\mathbb{F}_7,$ $\mathcal{S}=K^2$ and $\mathcal{M}$ the set of monomials of $K[x_1,x_2]$
whose exponents are
the points in the left picture of Example~\ref{19.08.17}.
The circled points in the Figure \ref{fig2} are the exponents of the monomials
that belong to the minimal generating set of $\mathcal{M}.$
\end{example}

\begin{figure} \label{fig2}
\begin{center}
\begin{tikzpicture}[scale=0.75]
\draw [-latex] (0,0) -- (7,0)node[right] {$K$};
\draw [dashed] (0,1)node[left]{1} -- (7,1)node[right] {};
\draw [dashed] (0,2)node[left]{2} -- (7,2)node[right] {};
\draw [dashed] (0,3)node[left]{3} -- (7,3)node[right] {};
\draw [dashed] (0,4)node[left]{4} -- (7,4)node[right] {};
\draw [dashed] (0,5)node[left]{5} -- (7,5)node[right] {};
\draw [dashed] (0,6)node[left]{6} -- (7,6)node[right] {};
\draw [-latex] (0,0)node[below left]{0} -- (0,7)node[above] {K};
\draw [dashed] (1,0)node[below]{1} -- (1,7)node[right] {};
\draw [dashed] (2,0)node[below]{2} -- (2,7)node[right] {};
\draw [dashed] (3,0)node[below]{3} -- (3,7)node[right] {};
\draw [dashed] (4,0)node[below]{4} -- (4,7)node[right] {};
\draw [dashed] (5,0)node[below]{5} -- (5,7)node[right] {};
\draw [dashed] (6,0)node[below]{6} -- (6,7)node[right] {};
\fill [color=red](0,0) {circle(.10cm)};
\fill [color=red](1,0) {circle(.10cm)};
\fill [color=red](2,0) {circle(.10cm)};
\fill [color=red](3,0) {circle(.10cm)};
\fill [color=red](4,0) {circle(.10cm)};
\fill [color=red](5,0) {circle(.10cm)};
\fill [color=red](0,1) {circle(.10cm)};
\fill [color=red](1,1) {circle(.10cm)};
\fill [color=red](2,1) {circle(.10cm)};
\fill [color=red](3,1) {circle(.10cm)};
\fill [color=red](4,1) {circle(.10cm)};
\fill [color=red](5,1) {circle(.10cm)};
\fill [color=red](0,2) {circle(.10cm)};
\fill [color=red](1,2) {circle(.10cm)};
\fill [color=red](2,2) {circle(.10cm)};
\fill [color=red](3,2) {circle(.10cm)};
\fill [color=red](4,2) {circle(.10cm)};
\fill [color=red](5,2) {circle(.10cm)};
\fill [color=red](0,3) {circle(.10cm)};
\fill [color=red](1,3) {circle(.10cm)};
\fill [color=red](2,3) {circle(.10cm)};
\fill [color=red](3,3) {circle(.10cm)};
\fill [color=red](4,3) {circle(.10cm)};
\fill [color=red](0,4) {circle(.10cm)};
\fill [color=red](1,4) {circle(.10cm)};
\fill [color=red](2,4) {circle(.10cm)};
\fill [color=red](3,4) {circle(.10cm)};
\fill [color=red](4,4) {circle(.10cm)};
\fill [color=red](0,5) {circle(.10cm)};
\fill [color=red](1,5) {circle(.10cm)};
\fill [color=red](2,5) {circle(.10cm)};
\fill [color=red](0,6) {circle(.10cm)};
\fill [color=red](1,6) {circle(.10cm)};
\fill [color=red](2,6) {circle(.10cm)};
\draw [color=blue, ultra thick](5,2) circle (7pt);
\draw [color=blue, ultra thick](4,4) circle (7pt);
\draw [color=blue, ultra thick](2,6) circle (7pt);
\end{tikzpicture}
\end{center} 
\caption{}
\end{figure}

From now on, $\mathcal{B}({\mathcal{M}})$ denotes the minimal generating set of $\mathcal{M}.$
We are going to describe properties of the code
$C(\mathcal{S},\mathcal{M})$ in terms of $\mathcal{B}({\mathcal{M}}).$
The following proposition explains how to find a generating set of $\mathcal{M}^c_{\mathcal{S}}$
in terms of $\mathcal{B}({\mathcal{M}}).$
\begin{proposition}\label{19.08.21}
Given a monomial $M=x_1^{a_1}\cdots x_m^{a_m}\in \mathcal{B}({\mathcal{M}}),$ define the monomials
$\displaystyle P(M):=\left\{\frac{x_1^{n_1-1}\cdots x_n^{n_m-1}}{x_i^{a_i-1}}:
 i\in[m], \text{ and } n_i-a_i-2\geq 0 \right\}.$
The set
\[  \operatorname{gcd} \left( P(M)\right)_{M\in \mathcal{B}({\mathcal{M}})}\]
is a generating set of $\mathcal{M}^c.$ The set
$\operatorname{gcd}$ is defined by induction, if $M_1, M_2$ and $M_3$ are elements
of $\mathcal{B}({\mathcal{M}}),$ then
\[\operatorname{gcd}(P(M_1),P(M_2),P(M_3))=
\operatorname{gcd}(\operatorname{gcd}(P(M_1),P(M_2)),P(M_3)),\] where 
$\operatorname{gcd}(P(M_1),P(M_2))=
\{\operatorname{gcd}(M_1^\prime,M_2^\prime): M_1^\prime\in M_1, M_2^\prime \in M_2\}.$
\end{proposition}
\begin{proof}
It is clear that for every monomial
$M=x_1^{a_1}\cdots x_m^{a_m}\in \mathcal{B}({\mathcal{M}})$ the set $P(M)$
is a minimal generating set for $\{ M\}^c.$ Given any two monomials $M_1$ and
$M_2,$ the set $\{ \operatorname{gcd}(M_1,M_2)\}$ is a minimal generating set for the
set of monomials that divide $M_1$ and $M_2,$ thus the result follows.
\end{proof}
It is important to note that the set
$\operatorname{gcd} \left( P(M)\right)_{M\in \mathcal{B}({\mathcal{M}})}$ from
Proposition~\ref{19.08.21} is not always a minimal generating set, as the following example shows.
\begin{example}\rm
Let $K=\mathbb{F}_7,$ $\mathcal{S}=K^2$ and $\mathcal{M}$ the set of monomials of $K[x_1,x_2]$
whose exponents are
the points in the left picture of Example~\ref{19.08.17}.
The circles in the picture of Example~\ref{19.08.20} are the exponents of the monomials
that belong to $ \mathcal{B}({\mathcal{M}}).$ The circles below are the exponents that belong
to \( \operatorname{gcd} \left( P(M)\right)_{M\in \mathcal{B}({\mathcal{M}})}\).
It is clear that it is not a minimal generating set.
\end{example}

\begin{figure}
\begin{center}
\begin{tikzpicture}[scale=0.75]
\draw [-latex] (6,6) -- (6,-0.5)node[below] {$K$};
\draw [dashed] (0,0)node[left]{} -- (6,0)node[right] {6};
\draw [dashed] (0,1)node[left]{} -- (6,1)node[right] {5};
\draw [dashed] (0,2)node[left]{} -- (6,2)node[right] {4};
\draw [dashed] (0,3)node[left]{} -- (6,3)node[right] {3};
\draw [dashed] (0,4)node[left]{} -- (6,4)node[right] {2};
\draw [dashed] (0,5)node[left]{} -- (6,5)node[right] {1};
\draw [dashed] (0,6)node[left]{} -- (6,6)node[right] {};
\draw [-latex] (6,6)node[above right]{0} -- (-0.5,6)node[left] {$K$};
\draw [dashed] (0,0)node[below]{} -- (0,6)node[above] {6};
\draw [dashed] (1,0)node[below]{} -- (1,6)node[above] {5};
\draw [dashed] (2,0)node[below]{} -- (2,6)node[above] {4};
\draw [dashed] (3,0)node[below]{} -- (3,6)node[above] {3};
\draw [dashed] (4,0)node[below]{} -- (4,6)node[above] {2};
\draw [dashed] (5,0)node[below]{} -- (5,6)node[above] {1};
\draw [dashed] (6,0)node[below]{} -- (6,6)node[above] {};
\fill [color=red](0,0) {circle(.10cm)};
\fill [color=red](1,0) {circle(.10cm)};
\fill [color=red](2,0) {circle(.10cm)};
\fill [color=red](3,0) {circle(.10cm)};
\fill [color=red](4,0) {circle(.10cm)};
\fill [color=red](5,0) {circle(.10cm)};
\fill [color=red](0,1) {circle(.10cm)};
\fill [color=red](1,1) {circle(.10cm)};
\fill [color=red](2,1) {circle(.10cm)};
\fill [color=red](3,1) {circle(.10cm)};
\fill [color=red](4,1) {circle(.10cm)};
\fill [color=red](5,1) {circle(.10cm)};
\fill [color=red](0,2) {circle(.10cm)};
\fill [color=red](1,2) {circle(.10cm)};
\fill [color=red](2,2) {circle(.10cm)};
\fill [color=red](3,2) {circle(.10cm)};
\fill [color=red](4,2) {circle(.10cm)};
\fill [color=red](5,2) {circle(.10cm)};
\fill [color=red](0,3) {circle(.10cm)};
\fill [color=red](1,3) {circle(.10cm)};
\fill [color=red](2,3) {circle(.10cm)};
\fill [color=red](3,3) {circle(.10cm)};
\fill [color=red](4,3) {circle(.10cm)};
\fill [color=red](0,4) {circle(.10cm)};
\fill [color=red](1,4) {circle(.10cm)};
\fill [color=red](2,4) {circle(.10cm)};
\fill [color=red](3,4) {circle(.10cm)};
\fill [color=red](4,4) {circle(.10cm)};
\fill [color=red](0,5) {circle(.10cm)};
\fill [color=red](1,5) {circle(.10cm)};
\fill [color=red](2,5) {circle(.10cm)};
\fill [color=red](0,6) {circle(.10cm)};
\fill [color=red](1,6) {circle(.10cm)};
\fill [color=red](2,6) {circle(.10cm)};

\draw [color=blue, ultra thick](3,5) circle (7pt);
\draw [color=blue, ultra thick](6,0) circle (7pt);
\draw [color=blue, ultra thick](6,5) circle (7pt);
\draw [color=blue, ultra thick](5,3) circle (7pt);
\end{tikzpicture}
\end{center}
\caption{}
\end{figure}

    Let $\mathcal{M}_1$ and $\mathcal{M}_1$ be two decreasing sets. Then
    $$\mathcal{B}(\mathcal{M}_1\cap \mathcal{M}_2)=\operatorname{gcd}(\mathcal{B}(\mathcal{M}_1),\mathcal{B}(\mathcal{M}_2))$$ and
    $$\mathcal{B}(\mathcal{M}_1\cup \mathcal{M}_2)=\mathcal{B}(\mathcal{M}_1)\cup\mathcal{B}(\mathcal{M}_2).$$To see this, note that if $M\in\mathcal{M}_1\cap\mathcal{M}_2$, then exists $M_1\in\mathcal{B}(\mathcal{M}_1)$ and $M_2\in\mathcal{B}(\mathcal{M}_2)$, such that $M|M_1$ and $M|M_2$. It follows that $$M|\operatorname{gcd}(M_1,M_2)\in\operatorname{gcd}\mathcal{B}(\mathcal{M}_1),\mathcal{B}(\mathcal{M}_2)).$$Therefore, $\mathcal{M}_1\cap\mathcal{M}_2\subset\operatorname{gcd}(\mathcal{B}(\mathcal{M}_1),\mathcal{B}(\mathcal{M}_2))$. The other containment is clear, as well as the claim for the union.

\begin{theorem}\label{19.08.18}
Consider a monomial-Cartesian code $C(\mathcal{S},\mathcal{M})$ as above. 
\begin{itemize}
\item[\rm (i)] The length of $C(\mathcal{S},\mathcal{M})$ is given by $\prod_{i=1}^m n_i.$
\item[\rm (ii)] The dimension of the code $C(\mathcal{S},\mathcal{M})$ is
\[\sum_{i=1}^{|\mathcal{B}({\mathcal{M}})|} \left( (-1)^{i-1} \sum_{T\in P_i}
\prod_{j=1}^m(t_j+1)
\right), \]
where $P_i \subseteqq \mathcal{B}({\mathcal{M}})$ are those subsets with $|P_i|=i$ and $(t_1,\ldots, t_m)$ is the exponent of $\operatorname{gcd}T$.
\item[\rm (iii)] The minimum distance of $C(\mathcal{S},\mathcal{M})$ is given by
\[ \min \left\{ \prod_{i=1}^m\left(n_i-a_i\right): x_1^{a_1}\cdots x_m^{a_m} \in
\mathcal{B}({\mathcal{M}})\right\}.\]
\end{itemize}
\end{theorem}
\begin{proof}
{\rm (i)} It is clear because $\displaystyle \prod_{i=1}^m n_i$ is the cardinality of $\mathcal{S}.$
{\rm (ii)} Given two monomials $M$ and $M^\prime,$ we see that
$\displaystyle \operatorname{gcd}(M,M^\prime)$ is the minimal generating set of the set of monomials
that divide to $M$ and also to $M^\prime.$ For any monomial
$\displaystyle M=x_1^{t_1}\cdots x_m^{t_m},$
$\displaystyle \prod_{j=1}^n(t_j+1)$ is the number of monomials that divide  $M.$ Thus the dimension
follows from the inclusion exclusion theorem.
{\rm (iii)} Let $\prec$ be the graded-lexicographical order and take
$\displaystyle f\in \operatorname{Span}_K\{M : M \in \mathcal{M} \}.$ If $M=x_1^{b_1}\cdots x_m^{b_m}$
is the leading monomial of $f$. Then \cite[Proposition 2.3]{carvalho} gives 
$\displaystyle |\operatorname{Supp}( \operatorname{ev}_{\mathcal{S}}f)| \geq
\prod_{i=1}^m\left(n_i-b_i\right).$
As $\mathcal{B}({\mathcal{M}})$ is a minimial generating set of $\mathcal{M},$ there exists
$M^\prime=x_1^{a_1}\cdots x_m^{a_m} \in \mathcal{B}({\mathcal{M}})$ such that
$M$ divides $M^\prime.$ Thus
$\displaystyle |\operatorname{Supp}( \operatorname{ev}_{\mathcal{S}}f)| \geq
\prod_{i=1}^m\left(n_i-a_i\right)$
and $\displaystyle d(C(\mathcal{S},\mathcal{M}))\geq \min \left\{ \prod_{i=1}^m\left(n_i-a_i\right):
x_1^{a_1}\cdots x_m^{a_m} \in \mathcal{B}({\mathcal{M}})\right\}.$
Assume for $i\in[m],$ $S_i=\left\{s_{i1},\ldots,s_{i n_i} \right\}.$
Let $x_1^{\alpha_1}\cdots x_m^{\alpha_m} \in \mathcal{B}({\mathcal{M}})$ such that
$\displaystyle \prod_{i=1}^m\left(n_i-\alpha_i\right)=\min \left\{ \prod_{i=1}^m\left(n_i-a_i\right):
x_1^{a_1}\cdots x_m^{a_m} \in \mathcal{B}({\mathcal{M}})\right\}.$
Define \(\displaystyle f_\alpha:=\prod_{i=1}^m\prod_{j=1}^{\alpha_i}\left(x_i-s_{ij} \right).\)
Since $\displaystyle |\operatorname{Supp}( \operatorname{ev}_{\mathcal{S}}f_\alpha)|=
\prod_{i=1}^m\left(n_i-a_i\right)$ and
$\displaystyle f_\alpha\in \operatorname{Span}_K\{M : M \in \mathcal{M} \}$ (as all monomials that
appear in $f_\alpha$ divide $\displaystyle x_1^{\alpha_1}\cdots x_m^{\alpha_m}$), then
we have $\displaystyle d (C(\mathcal{S},\mathcal{M}))\leq \min \left\{ \prod_{i=1}^m\left(n_i-a_i\right):
x_1^{a_1}\cdots x_m^{a_m} \in \mathcal{B}({\mathcal{M}})\right\}$ and the result follows.
\end{proof}
\begin{example}\rm
Let $K=\mathbb{F}_7,$ $\mathcal{S}=K^2$ and $\mathcal{M}$ be the set of monomials of $K[x_1,x_2]$
whose exponents are the points in the leftmost picture of Example~\ref{19.08.17}.
The length of the code is $49$, which is the total number of grid points in $\mathcal{S}.$
The dimension is $34$, which is the total number of points in the leftmost picture of Example~\ref{19.08.17}. 
The minimal generating set $\mathcal{B}({\mathcal{M}})$ is $\{x_1^2x_2^6, x_1^4x_2^4,x_1^5x_2^2\}.$
By Theorem~\ref{19.08.18}
$\displaystyle |\operatorname{Supp}( \operatorname{ev}_{\mathcal{S}}x^2y^6)|\geq 5,$ which is 
the number of grid points between the point $(2,6)$ and the point $(6,6).$ See the first picture
(from left to right) in Figure \ref{fig4}.
In a similar way
$\displaystyle |\operatorname{Supp}( \operatorname{ev}_{\mathcal{S}}x_1^4x_2^4)|\geq 9$
and $\displaystyle |\operatorname{Supp}( \operatorname{ev}_{\mathcal{S}}x_1^5x_2^2)|\geq 10.$
See second and third picture (from left to right) below.
As $\displaystyle \min \left\{ 5,9,10\right\}=5,$ the minimum distance 
$\displaystyle d (C(\mathcal{S},\mathcal{M}))$ is $5.$
\end{example}

\begin{figure} \label{fig4}
\begin{minipage}[t]{0.3\textwidth}
\begin{tikzpicture}[scale=0.45]
\draw [-latex] (0,0) -- (6.5,0)node[right] {};
\draw [dashed] (0,1)node[left]{1} -- (6,1)node[right] {};
\draw [dashed] (0,2)node[left]{2} -- (6,2)node[right] {};
\draw [dashed] (0,3)node[left]{3} -- (6,3)node[right] {};
\draw [dashed] (0,4)node[left]{4} -- (6,4)node[right] {};
\draw [dashed] (0,5)node[left]{5} -- (6,5)node[right] {};
\draw [dashed] (0,6)node[left]{6} -- (6,6)node[right] {};
\draw [-latex] (0,0)node[below left]{0} -- (0,6.5)node[above] {$K$};
\draw [dashed] (1,0)node[below]{1} -- (1,6)node[right] {};
\draw [dashed] (2,0)node[below]{2} -- (2,6)node[right] {};
\draw [dashed] (3,0)node[below]{3} -- (3,6)node[right] {};
\draw [dashed] (4,0)node[below]{4} -- (4,6)node[right] {};
\draw [dashed] (5,0)node[below]{5} -- (5,6)node[right] {};
\draw [dashed] (6,0)node[below]{6} -- (6,6)node[right] {};
\fill [color=red](0,0) {circle(.10cm)};
\fill [color=red](1,0) {circle(.10cm)};
\fill [color=red](2,0) {circle(.10cm)};
\fill [color=red](3,0) {circle(.10cm)};
\fill [color=red](4,0) {circle(.10cm)};
\fill [color=red](5,0) {circle(.10cm)};
\fill [color=red](0,1) {circle(.10cm)};
\fill [color=red](1,1) {circle(.10cm)};
\fill [color=red](2,1) {circle(.10cm)};
\fill [color=red](3,1) {circle(.10cm)};
\fill [color=red](4,1) {circle(.10cm)};
\fill [color=red](5,1) {circle(.10cm)};
\fill [color=red](0,2) {circle(.10cm)};
\fill [color=red](1,2) {circle(.10cm)};
\fill [color=red](2,2) {circle(.10cm)};
\fill [color=red](3,2) {circle(.10cm)};
\fill [color=red](4,2) {circle(.10cm)};
\fill [color=red](5,2) {circle(.10cm)};
\fill [color=red](0,3) {circle(.10cm)};
\fill [color=red](1,3) {circle(.10cm)};
\fill [color=red](2,3) {circle(.10cm)};
\fill [color=red](3,3) {circle(.10cm)};
\fill [color=red](4,3) {circle(.10cm)};
\fill [color=red](0,4) {circle(.10cm)};
\fill [color=red](1,4) {circle(.10cm)};
\fill [color=red](2,4) {circle(.10cm)};
\fill [color=red](3,4) {circle(.10cm)};
\fill [color=red](4,4) {circle(.10cm)};
\fill [color=red](0,5) {circle(.10cm)};
\fill [color=red](1,5) {circle(.10cm)};
\fill [color=red](2,5) {circle(.10cm)};
\fill [color=red](0,6) {circle(.10cm)};
\fill [color=red](1,6) {circle(.10cm)};
\fill [color=red](2,6) {circle(.10cm)};
\draw[draw=blue, ultra thick] (1.5,5.5) rectangle ++(5,1);
\draw (2,6)node{$\textcolor{blue}{\bm \times}$};
\draw (3,6)node{$\textcolor{blue}{\bm \times}$};
\draw (4,6)node{$\textcolor{blue}{\bm \times}$};
\draw (5,6)node{$\textcolor{blue}{\bm \times}$};
\draw (6,6)node{$\textcolor{blue}{\bm \times}$};
\end{tikzpicture}
\end{minipage}
\begin{minipage}[t]{0.3\textwidth}
\begin{tikzpicture}[scale=0.45]
\draw [-latex] (0,0) -- (6.5,0)node[right] {};
\draw [dashed] (0,1)node[left]{1} -- (6,1)node[right] {};
\draw [dashed] (0,2)node[left]{2} -- (6,2)node[right] {};
\draw [dashed] (0,3)node[left]{3} -- (6,3)node[right] {};
\draw [dashed] (0,4)node[left]{4} -- (6,4)node[right] {};
\draw [dashed] (0,5)node[left]{5} -- (6,5)node[right] {};
\draw [dashed] (0,6)node[left]{6} -- (6,6)node[right] {};
\draw [-latex] (0,0)node[below left]{0} -- (0,6.5)node[above] {$K$};
\draw [dashed] (1,0)node[below]{1} -- (1,6)node[right] {};
\draw [dashed] (2,0)node[below]{2} -- (2,6)node[right] {};
\draw [dashed] (3,0)node[below]{3} -- (3,6)node[right] {};
\draw [dashed] (4,0)node[below]{4} -- (4,6)node[right] {};
\draw [dashed] (5,0)node[below]{5} -- (5,6)node[right] {};
\draw [dashed] (6,0)node[below]{6} -- (6,6)node[right] {};
\fill [color=red](0,0) {circle(.10cm)};
\fill [color=red](1,0) {circle(.10cm)};
\fill [color=red](2,0) {circle(.10cm)};
\fill [color=red](3,0) {circle(.10cm)};
\fill [color=red](4,0) {circle(.10cm)};
\fill [color=red](5,0) {circle(.10cm)};
\fill [color=red](0,1) {circle(.10cm)};
\fill [color=red](1,1) {circle(.10cm)};
\fill [color=red](2,1) {circle(.10cm)};
\fill [color=red](3,1) {circle(.10cm)};
\fill [color=red](4,1) {circle(.10cm)};
\fill [color=red](5,1) {circle(.10cm)};
\fill [color=red](0,2) {circle(.10cm)};
\fill [color=red](1,2) {circle(.10cm)};
\fill [color=red](2,2) {circle(.10cm)};
\fill [color=red](3,2) {circle(.10cm)};
\fill [color=red](4,2) {circle(.10cm)};
\fill [color=red](5,2) {circle(.10cm)};
\fill [color=red](0,3) {circle(.10cm)};
\fill [color=red](1,3) {circle(.10cm)};
\fill [color=red](2,3) {circle(.10cm)};
\fill [color=red](3,3) {circle(.10cm)};
\fill [color=red](4,3) {circle(.10cm)};
\fill [color=red](0,4) {circle(.10cm)};
\fill [color=red](1,4) {circle(.10cm)};
\fill [color=red](2,4) {circle(.10cm)};
\fill [color=red](3,4) {circle(.10cm)};
\fill [color=red](4,4) {circle(.10cm)};
\fill [color=red](0,5) {circle(.10cm)};
\fill [color=red](1,5) {circle(.10cm)};
\fill [color=red](2,5) {circle(.10cm)};
\fill [color=red](0,6) {circle(.10cm)};
\fill [color=red](1,6) {circle(.10cm)};
\fill [color=red](2,6) {circle(.10cm)};
\draw[draw=blue, ultra thick] (3.5,3.5) rectangle ++(3,3);
\draw (4,4)node{$\textcolor{blue}{\bm \times}$};
\draw (5,4)node{$\textcolor{blue}{\bm \times}$};
\draw (6,4)node{$\textcolor{blue}{\bm \times}$};
\draw (4,5)node{$\textcolor{blue}{\bm \times}$};
\draw (5,5)node{$\textcolor{blue}{\bm \times}$};
\draw (6,5)node{$\textcolor{blue}{\bm \times}$};
\draw (4,6)node{$\textcolor{blue}{\bm \times}$};
\draw (5,6)node{$\textcolor{blue}{\bm \times}$};
\draw (6,6)node{$\textcolor{blue}{\bm \times}$};

\end{tikzpicture}
\end{minipage}
\begin{minipage}[t]{0.3\textwidth}
\begin{tikzpicture}[scale=0.45]
\draw [-latex] (0,0) -- (6.5,0)node[right] {};
\draw [dashed] (0,1)node[left]{1} -- (6,1)node[right] {};
\draw [dashed] (0,2)node[left]{2} -- (6,2)node[right] {};
\draw [dashed] (0,3)node[left]{3} -- (6,3)node[right] {};
\draw [dashed] (0,4)node[left]{4} -- (6,4)node[right] {};
\draw [dashed] (0,5)node[left]{5} -- (6,5)node[right] {};
\draw [dashed] (0,6)node[left]{6} -- (6,6)node[right] {};
\draw [-latex] (0,0)node[below left]{0} -- (0,6.5)node[above] {$K$};
\draw [dashed] (1,0)node[below]{1} -- (1,6)node[right] {};
\draw [dashed] (2,0)node[below]{2} -- (2,6)node[right] {};
\draw [dashed] (3,0)node[below]{3} -- (3,6)node[right] {};
\draw [dashed] (4,0)node[below]{4} -- (4,6)node[right] {};
\draw [dashed] (5,0)node[below]{5} -- (5,6)node[right] {};
\draw [dashed] (6,0)node[below]{6} -- (6,6)node[right] {};
\fill [color=red](0,0) {circle(.10cm)};
\fill [color=red](1,0) {circle(.10cm)};
\fill [color=red](2,0) {circle(.10cm)};
\fill [color=red](3,0) {circle(.10cm)};
\fill [color=red](4,0) {circle(.10cm)};
\fill [color=red](5,0) {circle(.10cm)};
\fill [color=red](0,1) {circle(.10cm)};
\fill [color=red](1,1) {circle(.10cm)};
\fill [color=red](2,1) {circle(.10cm)};
\fill [color=red](3,1) {circle(.10cm)};
\fill [color=red](4,1) {circle(.10cm)};
\fill [color=red](5,1) {circle(.10cm)};
\fill [color=red](0,2) {circle(.10cm)};
\fill [color=red](1,2) {circle(.10cm)};
\fill [color=red](2,2) {circle(.10cm)};
\fill [color=red](3,2) {circle(.10cm)};
\fill [color=red](4,2) {circle(.10cm)};
\fill [color=red](5,2) {circle(.10cm)};
\fill [color=red](0,3) {circle(.10cm)};
\fill [color=red](1,3) {circle(.10cm)};
\fill [color=red](2,3) {circle(.10cm)};
\fill [color=red](3,3) {circle(.10cm)};
\fill [color=red](4,3) {circle(.10cm)};
\fill [color=red](0,4) {circle(.10cm)};
\fill [color=red](1,4) {circle(.10cm)};
\fill [color=red](2,4) {circle(.10cm)};
\fill [color=red](3,4) {circle(.10cm)};
\fill [color=red](4,4) {circle(.10cm)};
\fill [color=red](0,5) {circle(.10cm)};
\fill [color=red](1,5) {circle(.10cm)};
\fill [color=red](2,5) {circle(.10cm)};
\fill [color=red](0,6) {circle(.10cm)};
\fill [color=red](1,6) {circle(.10cm)};
\fill [color=red](2,6) {circle(.10cm)};
\draw[draw=blue, ultra thick] (4.5,1.5) rectangle ++(2,5);
\draw (5,2)node{$\textcolor{blue}{\bm \times}$};
\draw (5,3)node{$\textcolor{blue}{\bm \times}$};
\draw (5,4)node{$\textcolor{blue}{\bm \times}$};
\draw (5,5)node{$\textcolor{blue}{\bm \times}$};
\draw (5,6)node{$\textcolor{blue}{\bm \times}$};
\draw (6,2)node{$\textcolor{blue}{\bm \times}$};
\draw (6,3)node{$\textcolor{blue}{\bm \times}$};
\draw (6,4)node{$\textcolor{blue}{\bm \times}$};
\draw (6,5)node{$\textcolor{blue}{\bm \times}$};
\draw (6,6)node{$\textcolor{blue}{\bm \times}$};
\end{tikzpicture}
\end{minipage}
\caption{}
\end{figure}

\section{Polar codes that are polar decreasing monomial-Cartesian codes} \label{polar_dec_section}
In this section we are going to represent families of polar codes in terms of the just defined decreasing monomial-Cartesian codes.
Throughout this section, we continue with the same notation so that $K$ represents the finite field $\mathbb{F}_q,$ $R:=K[x_1,\ldots,x_m]$ is the polynomial ring over $K$ in $m$ variables, $\mathcal{M}\subseteq R$ is a set of monomials that is {decreasing},  $S_1,\ldots,S_m$ are subsets of $K,$ $\mathcal{S}$ represents the Cartesian set $\mathcal{S}=S_1\times\cdots\times S_m,$ $n_i=|S_i|$ for $i\in [m],$ $n=|\mathcal{S}|$, and $C(\mathcal{S},\mathcal{M})$ represents the {decreasing monomial-Cartesian code} associated to $\mathcal{S}$ and $\mathcal{M}$.

We associate the following matrix to a set $S=\{a_1,\ldots,a_l\}\subseteq\mathbb{F}_q:$
\[T(S)=\begin{blockarray}{cccccc}
&a_1&a_2&\cdots&a_l\\
\begin{block}{c[ccccc]}
x^{l-1}&a_1^{l-1}&a_2^{l-1}&\cdots&a_l^{l-1}\\
\vdots&\vdots&\vdots&\ddots&\vdots\\
x&a_1&a_2&\cdots&a_l\\
1&1&1&\cdots&1\\
\end{block}
\end{blockarray}.\]
Notice that $T(S)$ is invertible, it has a non identity standard form and it is a generator matrix of the decreasing monomial-Cartesian code $C({S},\{1,\ldots,x^{l-1}\}).$ Take $S_1,S_2,\ldots,S_m\subseteq K$ and let $T_i=T(S_i)$. If $S_i=\{a_{i1},\ldots,a_{in_i}\}$, we can order the set $\mathcal{S}=S_1\times\cdots\times S_m$ with the order inherited from the lexicographical order; i.e.,          \[(a_{1j_1},\ldots,a_{mj_m})\preceq (a_{1h_1},\ldots,a_{mh_m})\Longleftrightarrow j_k<h_k,$$ where $$k=\min\{r\in\{1,\ldots,m\}\ |\ a_{rj_r}\neq a_{rh_r}\}.\]
Let $\mathcal{M}=\{x_1^{a_1}\cdots x_m^{a_m}\ |\ a_i\leq n_i-1,\ 1\leq i\leq m\}$ and order this set with the inverse lexicographical order. Then we have that 
\[G_m=B_m(T_1\otimes \cdots\otimes T_m),\]
where $B_m$ is the permutation matrix that sends the row $\displaystyle j=k_m+\sum_{i=1}^{m-1}k_in_{i+1}$ to the row $\displaystyle  j'=k_1+\sum_{i=2}^m k_in_{i-1},$ has as rows the evaluations $ev_\mathcal{S}$ of $\mathcal{M}$ in decreasing order.
\begin{example}\rm
    Let $\alpha$ be a primitive element of $\mathbb{F}_4$ and $S_1=\{0,1,\alpha\}$, $S_2=\mathbb{F}_4$. Then
    $$T_1=\begin{blockarray}{ccccc}
    &0&1&\alpha\\
    \begin{block}{c[cccc]}
    x^2&0&1&\alpha^2\\
    x&0&1&\alpha\\
    1&1&1&1\\
    \end{block}\end{blockarray}\hspace{1cm}T_2=\begin{blockarray}{cccccc}
    &0&1&\alpha&\alpha^2\\
    \begin{block}{c[ccccc]}
    y^3&0&1&1&1\\
    y^2&0&1&\alpha^2&\alpha\\
    y&0&1&\alpha&\alpha^2\\
    1&1&1&1&1\\
    \end{block}
    \end{blockarray}$$

    Therefore,
    $$G_2=\begin{blockarray}{cccccccccccccc}
    &00&01&0\alpha&0\alpha^2&10&11&1\alpha&1\alpha^2&\alpha0&\alpha1&\alpha\alpha&\alpha\alpha^2\\
    \begin{block}{c[ccccccccccccc]}
y^3x^2&0&      0&      0&      0&      0&      1&      1&      1&      0&      \alpha^2&      \alpha^2&      \alpha^2\\
y^3x&0&      0&      0&      0&      0&      1&      1&      1&      0&      \alpha&      \alpha&      \alpha\\
y^3&0&      1&      1&      1&      0&      1&      1&      1&      0&      1&      1&      1\\
y^2x^2&0&      0&      0&      0&      0&      1&      \alpha^2&      \alpha&      0&      \alpha^2&      \alpha&      1\\
y^2x&0&      0&      0&      0&      0&      1&      \alpha^2&      \alpha&      0&      \alpha&      1&      \alpha^2\\
y^2&0&      1&      \alpha^2&      \alpha&      0&      1&      \alpha^2&      \alpha&      0&      1&      \alpha^2&      \alpha\\
yx^2&0&      0&      0&      0&      0&      1&      \alpha&      \alpha^2&      0&      \alpha^2&      1&      \alpha\\
yx&0&      0&      0&      0&      0&      1&      \alpha&      \alpha^2&      0&      \alpha&      \alpha^2&      1\\
y&0&      1&      \alpha&      \alpha^2&      0&      1&      \alpha&      \alpha^2&      0&      1&      \alpha&      \alpha^2\\
x^2&0&      0&      0&      0&      1&      1&      1&      1&      \alpha^2&      \alpha^2&      \alpha^2&      \alpha^2\\
x&0&      0&      0&      0&      1&      1&      1&      1&      \alpha&      \alpha&      \alpha&      \alpha\\
1&1&      1&      1&      1&      1&      1&      1&      1&      1&      1&      1&      1\\
\end{block}
\end{blockarray}$$
\end{example}
Since each row of $G_m$ can be viewed as a monomial, by an abuse of notation, for a monomial $M\in\mathcal{M}$, we can write $I(M)$ and $Z(M)$ for $I\left(W_m^{(i)}\right)$ and $Z\left(W_m^{(i)}\right)$ respectively, where
\[Row_i G_m=ev_\mathcal{S}(M).\]

	In the usual polarization process, for a square matrix $G \in \mathbb{F}_q^{l \times l}$, the speed of polarization is measured via the exponent. This is defined as the number $E(G)$ such that for any channel $W$ the following hold.
	\begin{enumerate}
		\item[(i)] For any fixed $\beta<E(G)$, 
		$$\liminf_{n\rightarrow\infty} P[Z_n\leq 2^{-l^{n\beta}}]=I(W).$$
		
		\item[(ii)] For any fixed $\beta>E(G)$,
		$$\liminf_{n\rightarrow\infty}P[Z_n\geq 2^{-l^{n\beta}}]=1.$$
	\end{enumerate}
	
	Therefore, if $D_j=d(Row_j G,\langle Row_{j+1} G,\ldots,Row_l G\rangle)$, then 
	$$E(G)=\sum_{j=1}^l\frac{\ln D_j}{l\ln l}.$$

\begin{remark}\rm
A lower bound on the exponent of the matrix $G_m$ can be calculated directly from the set of monomials as follows: 
\begin{eqnarray*}
E(G_m)=\sum_{j=1}^l\frac{\ln D_j}{l\ln l}&=&\sum_{j=1}^l\frac{\ln d(Row_j G_m,\langle Row_{j+1} G_m,\ldots,Row_l G_m\rangle)}{l\ln l}\\
&=&\sum_{j=1}^l\frac{\ln d(Row_j G_m,Row_{j+1} G_m,\ldots,Row_l G_m)}{l\ln l}\\
&\geq &\frac{1}{l\ln l}\sum_{j=1}^l \ln \left[\min \left\{ \prod_{i=1}^m\left(n_i-a_i\right): x_1^{a_1}\cdots x_m^{a_m} \in
\mathcal{B}({\mathcal{M}^j})\right\} \right],
\end{eqnarray*}
where ${\mathcal{M}^j}$ represents the last $j$ monomials of the set $\mathcal{M}$ according to the inverse lexicographical order.
\end{remark}

\begin{remark}\rm
	In \cite{Vardohus} was proven that if $G_1$ and $G_2$ are two square non-singular matrices over $\mathbb{F}_q$, of sizes $l_1$ and $l_2$ respectively, then
	$$E(G_1\otimes G_2)=\frac{E(G_1)}{\log_{l_1}(l_1l_2)}+\frac{E(G_2)}{\log_{l_2}(l_1l_2)}.$$
	
	From this we have that
	\begin{equation}
	E(G_1\otimes\cdots\otimes G_s)=\sum_{j=1}^s\frac{E(G_j)}{\log_{l_j}(l_1\cdots l_s)}.\tag{$\ast$}
	\end{equation}
	
	Redefining in the obvious way the exponent for the multikernel process, in \cite{mk2} the authors proved that if $T_1,\ldots,T_s$ are kernels with size $l_1,\ldots,l_s$ and exponents $E_1,\ldots,E_s$ are used to construct a multikernel polar code in which each $T_j$ appears with frequency $p_j$ on $G_N$ (the Kronecker product of these matrices) as $N\rightarrow\infty$, then the exponent of the multikernel process is
	$$E=\sum_{j=1}^s \frac{p_j\log_2(l_j)}{\sum_{k=1}^s p_k\log_2(l_k)}E_j,$$ which results to be
	$$E=\lim_{N\rightarrow\infty} E(G_N),$$ because of $(\ast)$. 
	
	For the case we are working on, each $T_i$ has size $l_i\leq q$ and we know $E(T_i)=\frac{\ln l_i!}{l_i\ln l_i}$, which is the best exponent we can get over all the matrices of size $l_i$. Given that for $G_m=B_m(T_1\otimes\cdots\otimes T_m)$ there exists a matrix permutation $P$ such that $G_mP=T_1\otimes \cdots\otimes T_m$ and     			$$E(G_m)=E(T_1\otimes\cdots\otimes T_m)=\sum_{i=1}^m \frac{E(T_i)}{\log_{l_i}(l_1\cdots l_m)}.$$
	
	Therefore, for any other matrix $G=M_1\otimes\cdots\otimes M_m$, such that $M_i$ is a square matrix of size $l_i$, $E(G)\leq E(G_m)$. Even more, for any sequence $\{T_i\}_{i=1}^\infty$, where $T_i$ is associated to a subset from $\mathbb{F}_q$, we have
	$$\lim_{n\rightarrow\infty}\sum_{k=1}^n\frac{E(T_k)}{\ln(l_1\cdots l_k)}\leq\frac{\ln q!}{q\ln q}.$$
	
	This suggests that the result exposed in \cite{mk2} could be generalized for this case.
\end{remark}

Let us continue with the description of information sets of polar codes.
The following monomial order is inspired by the order introduced in~\cite{Bardet}. They coincide when $K=\mathbb{F}_2$ and $S_1=\cdots=S_m=\mathbb{F}_2.$

The following definition is the key to define polar decreasing monomial-Cartesian codes in in terms of decreasing monomial-Cartesian codes.
\begin{definition}\rm\label{19.12.05}
Let $S_1,\ldots,S_m \subseteq K$ and $M, M', \tilde{M}, \tilde{M}'$ be monomials in $R.$ Define the monomial order $\trianglelefteq$ in $R$ as follows. 
    \begin{enumerate}
        \item[(i)] If $M'|M$, then $M'\trianglelefteq M$. 
        
        \item[(ii)] Suppose $S_{i_1}=\cdots=S_{i_r}$. Given 
        $\{j_1,\ldots,j_s\},\{h_1,\ldots,h_s\}\subseteq\{i_1,\ldots,i_r\}$ with $j_l <j_{l+1}$, $h_l < h_{l+1}$, fir $l=1, \dots, s-1$, and $i_l<i_{l+1}$ for $l=1, \dots, r-1$,
        \[x_{j_1}^{a_1}\cdots x_{j_s}^{a_s}\trianglelefteq x_{h_1}^{a_1}\cdots x_{h_s}^{a_s}\] if and only if $j_k\leq h_k$ for all $1\leq k\leq s$.
        
        \item[(iii)] Let $1\leq k \leq m-1.$ For $M,M'\in K[x_1,\ldots,x_k], \tilde{M},\tilde{M}'\in K[x_{k+1},\ldots,x_m]$, if $M\trianglelefteq M'$ and
        $\tilde{M}\trianglelefteq\tilde{M}',$ then
        \[M\tilde{M}\trianglelefteq M'\tilde{M}'.\]
        
    \end{enumerate}
\begin{example}\rm
Over $\mathbb{F}_5$, take $S_1=S_2=\{0,1,2\}$ and $S_3=\mathbb{F}_5$. As $x_3|x_2^2x_3$, then $x_3\trianglelefteq x_2^2x_3.$ Since $S_2=S_1$, then $x_1\trianglelefteq x_2$. Finally, since $x_1\trianglelefteq x_2,$ then $x_1x_3\trianglelefteq x_2x_3.$
\end{example}
A {\bf polar decreasing monomial-Cartesian code}  is a decreasing monomial-Cartesian code $C(\mathcal{S},\mathcal{M}),$ where $\mathcal{M}$ is closed under $\trianglelefteq.$
\end{definition}
\begin{lemma}\cite[Propositions 15, 20 and 27]{Vardohus}
Let $\{T_i\}_{i=1}^m$ be the sequence of the associated matrices to a sequence of sets $\{S_i\}_{i=1}^m$ of $K$. Let $G_n=B_n(T_1\otimes\cdots\otimes T_m)$
as before. If $M,M'\in K[x_1,\ldots,x_m]$ and $M\trianglelefteq M'$, then
\[I(M)\geq I(M') \qquad  \text{ and } \qquad Z(M)\leq Z(M').\]
\end{lemma}
If we represent the set $\mathcal{A}_m$ given in Definition~\ref{19.12.07} not as indexes of rows, but as monomials, then we have the next characterization of $\mathcal{A}_m$.
\begin{proposition}
Let $\{T_i\}_{i=1}^m$ be the sequence of  matrices associated with a sequence of sets $\{S_i\}_{i=1}^m$ of $K$. Let $\mathcal{A}_m$ be an information set given in Definition~\ref{19.12.07} by the sequence $\{T_i\}_{i=1}^m$. If $M\in\mathcal{A}_m$ and $M'\trianglelefteq M$, then $M'\in\mathcal{A}_m$.
\end{proposition}

The immediate consequence is that any polar code constructed from a sequence of subsets of $K$ is a polar decreasing monomial-Cartesian code.

\begin{theorem}\label{19.12.04}
Let $\{S_i\}_{i=1}^\infty$ be a sequence of subsets of $\mathbb{F}_q$ and let $\{T_i\}_{i=1}^\infty$ be the sequence of associated matrices. Then $\{T_i\}_{i=1}^\infty$ polarizes any SOF channel and a polar code $C_{\mathcal{A}_m}$ given in Definition~\ref{19.12.07} is a polar decreasing monomial-Cartesian code.
\end{theorem}
In \cite{isa}, the authors analyzed through a different order the information set for polar codes constructed with $G_A$.We could find a set of monomials $\mathcal{M}'$ such that 
\[\mathcal{A}_n=\{M\ |\ M\trianglelefteq M',\ M'\in\mathcal{M}'\}.\]
If we choose $\mathcal{M}'$ to be minimal, then we can called it a generating set of $\mathcal{A}_n$ as in \cite{isa}. However, since $\trianglelefteq$ considers more than just the divisibility, if $\mathcal{B}(\mathcal{A}_n)$ is the generating set in the sense of Definition \ref{gset}, $\mathcal{B}(\mathcal{A}_n)$ could be bigger than $\mathcal{M}'$. For example, consider $S_1=S_3=\{0,1,2\}\subset\mathbb{F}_5$ and $S_2=\mathbb{F}_5$. If we take
\[\mathcal{A}_3=\{x_2^2x_3,x_2x_3,x_3,x_2^2,x_2,x_1,1\},\]
a minimal basis respect to $\trianglelefteq$ is $\{x_2^2x_3\}$, but $\mathcal{B}(\mathcal{A}_3)=\{x_2^2x_3,x_1\}$.

\section{Conclusion} \label{conclusion}
In this paper we prove that if a sequence of invertible matrices $\{T_i\}_{i=1}^\infty$ over an arbitrary field $\mathbb{F}_q$ has the property that every $T_i$ has a non identity standard form, then the sequence $\{T_i\}_{i=1}^\infty$ polarizes any symmetric over the field channel (SOF channel) $W$. Given a sequence $\{T_i\}_{i=1}^\infty$ that polarizes, and a natural number $m,$ we define a polar code as the space generated by some rows of the matrix $G_m,$ where $G_m$ is defined inductively taking $G_1=T_1$ and for $m \geq 2$,
\[G_m=\begin{bmatrix}
G_{m-1}\otimes Row_1T_m\\
G_{m-1}\otimes Row_2T_m\\
\vdots\\
G_{m-1}\otimes Row_{l_m}T_m\end{bmatrix}.\]
Given a set of monomials $\mathcal{M}$ that is closed under divisibility and a Cartesian product $\mathcal{S},$ we used the theory of evaluation codes to study decreasing monomial-Cartesian codes, which are defined by evaluating the monomials of $\mathcal{M}$ over the set $\mathcal{S}.$ We prove that the dual of a decreasing monomial-Cartesian code is a code of the same type. Then we describe its basic parameters in terms of the minimal generating set of $\mathcal{M}$. These codes are important because when the set $\mathcal{M}$ is also closed under the monomial order $\trianglelefteq,$ then the evaluation code is called polar decreasing monomial-Cartesian code. Strengthening the symmetry required of the channel and using matrices associated with subsets of a finite field $\mathbb{F}_q$, we prove that families of polar codes with multiple kernels can be viewed as decreasing monomial-Cartesian codes and therefore any information set $\mathcal{A}_n$ can be described in a similar way, offering an unified treatment for this kind of codes. 

\section*{Acknowledgments}
The first and fourth author are partially supported by SIP-IPN, project 20195717, and CONACyT. The third author is partially supported by NSF DMS-1855136.

\end{document}